\documentclass[11pt]{article}

\usepackage{titling}

\usepackage{amsmath,amssymb,amsthm,verbatim}
\usepackage{mathabx}
\usepackage{caption}
\usepackage{color}
\usepackage{palatino}
\usepackage{mathpazo}
\usepackage{stmaryrd}
\usepackage[margin=1in]{geometry}
\usepackage{mathtools}
\mathtoolsset{centercolon}
\usepackage{thm-restate}
\usepackage{hyperref}
\definecolor{darkred}  {rgb}{0.5,0,0}
\definecolor{darkblue} {rgb}{0,0,0.5}
\definecolor{darkgreen}{rgb}{0,0.5,0}
\hypersetup{
  pdftitle = {Cryptographic Reversibility Notes},
  pdfauthor = {Eric Chitambar},
  colorlinks = true,
  urlcolor  = blue,         
  linkcolor = darkblue,     
  citecolor = darkgreen,    
  filecolor = darkred       
}

\newtheorem{theorem}{Theorem}
\newtheorem*{theorem*}{Theorem}
\newtheorem{proposition}{Proposition}
\newtheorem*{proposition*}{Proposition}
\newtheorem{lemma}{Lemma}

\newtheorem{corollary}{Corollary}
\theoremstyle{definition}

\newtheorem*{remark}{Remark}
\newtheorem*{example}{Example}

\newcommand{\mc}[1]{\mathcal{#1}}

\newcommand{\KAB}{\overrightarrow{K}(X:Y||Z)}
\newcommand{\KBA}{\overleftarrow{K}(X:Y||Z)}
\newcommand{\KCR}{K^{c.r.}(X:Y||Z)}
\title{Distributions Attaining Secret Key at a Rate of the Conditional Mutual Information}

\author{Eric Chitambar$^1$, \quad Ben Fortescue$^1$, \quad  Min-Hsiu Hsieh$^2$
\\[4mm]
\textit{$^1$ Department of Physics and Astronomy, Southern Illinois University,}\\ 
\textit{Carbondale, Illinois 62901, USA}\\
\textit{$^2$ Centre for Quantum Computation \& Intelligent Systems (QCIS),}\\
\textit{Faculty of Engineering and Information Technology (FEIT),}\\
\textit{University of Technology Sydney (UTS), NSW 2007, Australia}}

\date{\today}

\begin{document}
\maketitle

\vspace{-1cm}

\begin{abstract}
In this paper we consider the problem of extracting secret key from an eavesdropped source $p_{XYZ}$ at a rate given by the conditional mutual information.  We investigate this question under three different scenarios: (i) Alice ($X$) and Bob ($Y$) are unable to communicate but share common randomness with the eavesdropper Eve ($Z$), (ii) Alice and Bob are allowed one-way public communication, and (iii) Alice and Bob are allowed two-way public communication.  Distributions having a key rate of the conditional mutual information are precisely those in which a ``helping'' Eve offers Alice and Bob no greater advantage for obtaining secret key than a fully adversarial one.  For each of the above scenarios, strong necessary conditions are derived on the structure of distributions attaining a secret key rate of $I(X:Y|Z)$.  In obtaining our results, we completely solve the problem of secret key distillation under scenario (i) and identify $H(S|Z)$ to be the optimal key rate using shared randomness, where $S$ is the G\'{a}cs-K\"{o}rner Common Information.   We thus provide an operational interpretation of the conditional G\'{a}cs-K\"{o}rner Common Information.  Additionally, we introduce simple example distributions in which the rate $I(X:Y|Z)$ is achievable if and only if two-way communication is allowed.  
\end{abstract}
\vspace{-.375cm}

\section{Introduction}

A basic information-processing task involves the exchange of secret information between Alice ($X$) and Bob ($Y$) in the presence of an eavesdropper, Eve ($E$).  If Alice and Bob have some pre-established key that is secret from Eve, then any future message $M$ can be transmitted using the key as a one-time pad.  Thus, the problem of private communication can be reduced to the problem of \textit{secret key distillation}, which studies the extraction of secret key $\Phi_{XY}\cdot q_Z$ from some initial tripartite correlation $p_{XYZ}$.  Here, $\Phi_{XY}$ is a perfectly correlated bit and $q_Z$ is an arbitrary distribution.  Often, the correlations $p_{XYZ}$ are presented as a many-copy source $p_{XYZ}^n$, and Alice and Bob wish to know the optimal rate of secret bits per copy that they can distill from this source.     

It turns out that Alice and Bob can often enhance their distillation capabilities by openly disclosing some information about $X$ and $Y$ through public communication \cite{Ahlswede-1993a, Maurer-1993a}.  In general, Alice and Bob's communication schemes can be interactive with one round of communication depending on what particular messages were broadcasted in previous rounds.  Such interactive protocols are known to generate higher key rates than non-interactive protocols, at least in the absence of ``noisy'' local processing by Alice and Bob \cite{Maurer-1993a}.  Thus, for a given distribution $p_{XYZ}$, one obtains a hierarchy of key rates pertaining to the respective scenarios of no communication, one-way communication, and two-way (interactive) communication.  It is also possible to consider no-communication scenarios in which Alice and Bob have access to some publically shared randomness that is uncorrelated with their primary source $p_{XYZ}$.  Clearly publically shared randomness is a weaker resource than public communication since the latter is able to generate the former.  However, below we will prove even stronger that publically shared randomness offers no advantage whatsoever for secret key distillation.  

For the one-way communication scenario, a single-letter characterization of the key rate has been proven by Ahlswede and Csisz\'{a}r \cite{Ahlswede-1993a}.  When the unidirectional communication is from Alice to Bob, we denote the key rate by $\KAB$, while $\KBA$ denotes the rate when communication is from Bob to Alice only.  No formula is known for the two-way key rate of a given distribution, which we denote by $K(X:Y||Z)$, and the complexity of protocols utilizing interactive communication makes computing this a highly challenging open problem.  

In the special case of an uncorrelated Eve in $p_{XYZ}$, the key rate is given by the mutual information $I(X:Y)$, and this can be achieved using one-way communication.  For more general distributions in which Eve possesses some side information of $XY$, the conditional mutual information $I(X:Y|Z)$ is a known upper bound for the key rate under two-way communication \cite{Ahlswede-1993a, Maurer-1993a}.  In general this bound is not tight \cite{Maurer-1999a}.  Rather, the conditional mutual information quantifies the key rate when Eve helps Alice and Bob by broadcasting her variable $Z$.  Key obtained by a helping Eve is also known as \textit{private key} \cite{Csiszar-2000a}, and private key is still secret from Eve even though she helps Alice and Bob obtain it.  The relevance of private key naturally arises in situations where Eve functions as a central server who helps establish secret correlations between Alice and Bob.  Thus, distributions with a secret key rate equaling the private key rate of $I(X:Y|Z)$  are precisely those in which nothing is gained by a helping Eve.

The objective of this paper is to investigate the types of distributions for which $I(X:Y|Z)$ is indeed an achievable secret key rate.  This will be considered under the scenarios of (i) publically shared randomness but no communication, (ii) one-way communication, and (iii) two-way communication.  A full solution to the problem would involve a structural characterization of the distributions $p_{XYZ}$ whose key rates are $I(X:Y|Z)$.  We are able to fully achieve this only for the no-communication setting, but we nevertheless derive strong necessary conditions for both the one-way and the two-way scenarios.  In the case of one-way communication, our condition makes use of the key-rate formula derived by Ahlswede and Csisz\'{a}r.  For the statement of this formula, recall that three variables $A$, $B$, and $C$ satisfy the Markov chain $A-B-C$ if $C$ is conditionally independent of $A$ given $B$; i.e. $p(c|b,a)=p(c|b)$ for letters in the range of $A$, $B$, and $C$.  Then,
\begin{lemma}[\cite{Ahlswede-1993a}]
\label{Lem:AC-Lemma}
For distribution $p_{XYZ}$, 
\begin{equation}
\label{Eq:One-way Rate}
\KAB=\max_{KU|XYZ} I(K:Y|U)-I(K:Z|U),
\end{equation}
where the maximization is taken over all auxiliary variables $K$ and $U$ satisfying the Markov chain $KU-X-YZ$, with $K$ and $U$ ranging over sets of size no greater than $|\mc{X}|+1$.  In particular,
\begin{equation}
\overrightarrow{K}(X:Y||Z)\geq I(X:Y)-I(X:Z).
\end{equation}
\end{lemma}

In this paper, we consider when variables $KU$ can be found that satisfy both $KU-X-YZ$ and $I(K;Y|U)-I(K;Z|U)=I(X:Y|Z)$.  Theorem \ref{Thm:One-way-conditions1} below offers a necessary condition on the structure of distributions for which this is possible.  Turning to the scenario of two-way communication, we utilize the well-known intrinsic information upper bound on $K(X:Y||Z)$.  For distribution $p_{XYZ}$, its intrinsic information is given by
\begin{equation}
\label{Eq:IntrinsicInfo}
I(X:Y\downarrow Z):=\min_{\overline{Z}|Z} I(X:Y|\overline{Z})
\end{equation}
where the minimization is taken over over all auxiliary variables $\overline{Z}$ satisfying $XY-Z-\overline{Z}$, with $\overline{Z}$ having the same range as $Z$ \cite{Christandl-2003a}.  Thus, the intrinsic information is the smallest conditional mutual information achievable after Eve processes her variable $Z$.  The intrinsic information satisfies $K(X:Y||Z)\leq I(X:Y\downarrow Z)$.  In Theorem \ref{Thm:DistOpt1} below, we identify a large class of distributions for which a channel $\overline{Z}|Z$ can be found satisfying $I(X:Y|\overline{Z})<I(X:Y|Z)$.  This allows us to derive a necessary condition on distributions having $K(X:Y||Z)=I(X:Y|Z)$.  

A brief summary of our results is the following:
\begin{itemize}
\item For publically shared randomness with no communication, we identify $H(J_{XY}|Z)$ as the secret key rate, where $J_{XY}$ is the G\'{a}cs-K\"{o}rner Common Information of Alice and Bob's marginal distribution $p_{XY}$.  Moreover, this rate is achievable without using shared randomness.  Using this result, the structure of distributions attaining $I(X:Y|Z)$ can easily be characterized.

\item When one-way communication is permitted between Alice and Bob, we show that the distribution $p_{XYZ}$ must satisfy a certain ``block-like'' structure in order to obtain the key rate $I(X:Y|Z)$.  Specifically, given some outcome $z$ of Eve, if there exists collections of events $\mc{X}_0$ and $\mc{Y}_0$ for Alice and Bob respectively that satisfy $p(\mc{Y}_0|\mc{X}_0,z)=p(\mc{X}_0|\mc{Y}_0,z)=1$, then $p(\mc{Y}_0|\mc{X}_0)=p(\mc{X}_0|\mc{Y}_0)=1$; i.e. the conditional probabilities hold regardless of Eve's outcome. 

\item For key distillation with two-way communication, we show that distributions attaining a key rate of $I(X:Y|Z)$ must also satisfy a certain type of uniformity similar to the one-way case.  One special class of distributions our necessary condition applies to are those obtained by mixing a perfectly correlated distribution $p_{XY}$ with an uncorrelated one such that the marginal distributions have the same range and such that Eve's variable $Z$ specifies which one of the distributions Alice and Bob hold.  We show that unless either Alice or Bob can likewise identify the distribution from his or her variable, a key rate of $I(X:Y|Z)$ is unattainable.

\item We construct distributions in which a distillation rate of $I(X:Y|Z)$ is unachievable when the communication is restricted from Alice to Bob, and yet it becomes achievable if the communication direction is from Bob to Alice.  We further provide an example when $I(X:Y|Z)$ is achievable only if two-way communication is used.  To our knowledge, these are the first known examples rigorously demonstrating such communication dependency for optimal key distillation.  We then turn to the difference between single-party key extraction versus shared key extraction by public communication.  We completely characterize the distributions in which the latter can be accomplished at the same rate as the former.
\end{itemize}

Before presenting these results in greater detail, we begin in Section \ref{Sect:Defns} with a more precise overview of the key rates studied in this paper.  In Section \ref{Sect:Gacs-Korner}, we then present the G\'{a}cs-K\"{o}rner Common Information and prove some basic properties.  Section \ref{Sect:Results} contains our main results, with longer proofs postponed to the appendix.  Finally, Section \ref{Sect:Conclusions} offers some concluding remarks.

\section{Definitions}

\label{Sect:Defns}

Let us review the relevant definitions of secret key rate under various communication scenarios.  We consider random variables $X$, $Y$ and $Z$ ranging over finite alphabets $\mc{X}$, $\mc{Y}$, and $\mc{Z}$ respectively.  For a general distribution $q$, we say its support (denoted by $supp[q]$) is the collection of $x$ such that $q(x)>0$.   In all distillation tasks, we assume that Alice and Bob each have access to one part of an i.i.d. (identical and independently distributed) source $XYZ$ whose distribution is $p_{XYZ}$.  Hence, after $n$ realizations of the source, $X^n$, $Y^n$ and $Z^n$ belong to Alice, Bob, and Eve respectively.  In addition, Alice and Bob each possess a local random variable, $Q_A$ and $Q_B$ respectively, which are mutually independent from each other and from $X^nY^nZ^n$.  This allows them to introduce local randomness into their processing of $X^nY^n$.  

We first turn to the most restrictive scenario, which is key distillation using publicly shared randomness.  The \textit{common randomness (c.r.) key rate} of $X$, $Y$, and $Z$, denoted by $K^{c.r.}(X:Y||Z)$, is defined to be the largest $R$ such that for every $\epsilon>0$, there is an integer $N$ such that $n\geq N$ implies the existence of (a) a random variable $W$ independent of $X^nY^nZ^n$ and ranging over some set $\mc{W}$, (b) a random variable $K$ ranging over some set $\mc{K}$, and (c) a pair of mappings $f(X^n,Q_A,W)$ and $g(Y^n,Q_B,W)$ for which
\begin{enumerate}
\item[(i)] $Pr[f=g=K]>1-\epsilon$;
\item[(ii)] $\log|\mc{K}|-H(K|Z^nW)<\epsilon$;
\item[(iii)] $\frac{1}{n}\log|\mc{K}|\geq R$.
\end{enumerate}

We next move to the more general scenario of when Alice and Bob are allowed to engage in public communication.  A \textit{local operations and public communication} (LOPC) protocol consists of a sequence of public communication exchanges between Alice and Bob.  The $i^{th}$ message exchanged between them is described by the variable $M_i$.  If Alice (resp. Bob) is the broadcasting party in round $i$, then $M_i$ is a function of $X^n$ and $Q_A$ (resp. $Y^n$ and $Q_B$) as well as the previous messages $(M_1,M_2,\cdots,M_{i-1})$.  The protocol is one-way if there is only one round of a message exchange.  

For distribution $p_{XYZ}$, the \textit{Alice-to-Bob secret key rate} $\KAB$ is the largest $R$ that satisfies the above three conditions except with $W$ being replaced by some message $M$ that is generated by Alice and therefore a function of ($X^n$, $Q_A$).  We can likewise define the Bob-to-Alice key rate $\KBA$.  The \textit{(two-way) secret key rate} of $X$ and $Y$ given $Z$, denoted by $K(X:Y||Z)$, is defined analogously except with $M=(M_1,M_2,\cdots,M_{r})$ being any random variable generated by an LOPC protocol \cite{Maurer-1993a, Ahlswede-1993a}.  The key rates satisfy the obvious relationship:
\begin{equation}
K^{c.r.}(X:Y||Z)\leq\{\KAB,\KBA\} \leq K(X:Y||Z).
\end{equation}

\section{The G\'{a}cs-K\"{o}rner Common Information}

\label{Sect:Gacs-Korner}

In this section, we introduce the G\'{a}cs-K\"{o}rner Common Information.  For every pair of random variables $XY$, there exists a \textit{maximal} common variable $J_{XY}$ in the sense that $J_{XY}$ is a function of both $X$ and $Y$, and any other such common function of both $X$ and $Y$ is itself a function of $J_{XY}$.  Hence, up to relabeling, the variable $J_{XY}$ is unique for each distribution $p_{XY}$.   In terms of its structure, a distribution $p_{XY}$ can always be decomposed as
\begin{equation}
\label{Eq:CommonInfoDecomp}
p(x,y)=\sum_{J_{XY}=j}p(x,y|j)p(j),
\end{equation}
where for any $x,x'\in\mc{X}$ and $y,y'\in\mc{Y}$, the conditional distributions satisfy $p(x,y|j)p(x,y'|j')=0$ and $p(x,y|j)p(x',y|j')=0$ if $j\not=j'$.  G\'{a}cs and K\"{o}rner identify $H(J_{XY})$ as the common information of $XY$ \cite{Gacs-1973a}.

It is instructive to rigorously prove the statements of the preceding paragraph.  A \textit{common partitioning of length $t$} for $XY$ are pairs of subsets $(\mc{X}_i,\mc{Y}_i)_{i=1}^t$ such that 
\begin{itemize}
\item[(i)] $\mc{X}_i\cap\mc{X}_j=\mc{Y}_i\cap \mc{Y}_j=\emptyset$ for $i\not=j$, 
\item[(ii)] $p(\mc{X}_i|\mc{Y}_j)=p(\mc{Y}_i|\mc{X}_j)=\delta_{ij}$, and  
\item[(iii)] if $(x,y)\in\mc{X}_i\times \mc{Y}_i$ for some $i$, then $p_X(x)p_Y(y)>0$. 
\end{itemize}
For a given common partitioning, we refer to the subsets $\mc{X}_i\times \mc{Y}_i$ as the ``blocks'' of the partitioning.  The subscript $i$ merely serves to label the different blocks, and for any fixed labeling, we associate a random variable $C(X,Y)$ such that $C(x,y)=i$ if $(x,y)\in\mc{X}_i\times\mc{Y}_i$.  Note that each party can determine the value of $J$ from their local information, and it is therefore called a \textit{common function} of $X$ and $Y$.  A \textit{maximal common partitioning} is a common partitioning of greatest length.  The following proposition is proven in the appendix.  
\begin{proposition}
\label{Prop:Partition-Unique}
\begin{itemize}
\item[{}]
\item[(a)]  Every pair of finite random variables $XY$ has a unique maximal common partitioning, which we denote by $J_{XY}$,
\item[(b)]  Variable $J_{XY}$ satisfies
\[H(J_{XY})=\max_K\{H(K):0=H(K|X)=H(K|Y)\}\]
iff $J_{XY}$ is a common function for the maximal common partitioning of $XY$.
\item[(c)] If $f(X)=g(Y)=C$ is any other common function of $X$ and $Y$, then $C(J_{XY})$.
\end{itemize}

\end{proposition}

With property (a), we can speak unambiguously of \textit{the} maximal common partitioning of a distribution $p_{XY}$.  Consequently the variable $J_{XY}$ is unique up to a relabeling of its range.  The following proposition provides a useful characterization of values $x$ and $x'$ that belong to the same block in a maximal common partitioning.
\begin{proposition}
\label{Prop:Ergodic}
If $J_{XY}(x)=J_{XY}(x')$ for $x,x'\in J_{XY}$, then there exists a sequence of values
\[xy_1x_1y_2x_2\cdots y_n x'\]
such that $p(x,y_1)p(y_1,x_1)p(x_1,y_2)\cdots p(y_n,x')>0$.
\end{proposition}
\begin{proof}
See the appendix as well as \cite{Gacs-1973a}.
\end{proof}

\section{Results}  

\label{Sect:Results}

\subsection{Key Distillation Using Auxiliary Public Randomness}

\label{Sect:DistCommonRand}

The G\'{a}cs and K\"{o}rner Common Information plays a central role in the problem of key distillation with no communication.  To see a preliminary connection, we recall an operational interpretation of $H(J_{XY})$ that G\'{a}cs and K\"{o}rner prove in Ref. \cite{Gacs-1973a}.  The task involves Alice and Bob constructing faithful encodings of their respective sources $X$ and $Y$, and $H(J_{XY})$ quantifies the asymptotic average sequence-length of codewords per copy such that both Alice and Bob's encodings output matching codewords with high probability over this sequence \cite{Gacs-1973a}.  

For the task of key distillation, Alice and Bob are likewise trying to convert their sources into matching sequences of optimal length.  However, the key distillation problem is different in two ways.  On the one hand there is the additional constraint that the common sequence should be nearly uncorrelated from Eve.  On the other hand, unlike the  G\'{a}cs-K\"{o}rner problem, it is not required that these sequences belong to faithful encodings of the sources $X$ and $Y$.  Nevertheless, we find that $H(J_{XY}|Z)$ quantifies the distillable key when Alice and Bob are unable to communicate with one another.  This is also the rate even if Alice and Bob have access to auxillary public randomness which is uncorrelated with their primary distribution.  
\begin{theorem}
\label{Thm:CR-Theorem}
$K^{c.r.}(X:Y||Z)=H(J_{XY}|Z)$.  Moreover, $H(J_{XY}|Z)$ is achievable with no additional common randomness.
\end{theorem}
\begin{proof}
\noindent\textbf{Achievability:}  
We will prove that $H(J_{XY}|Z)$ is an achievable rate without any auxiliary shared public randomness (i.e. $W$ is constant).  For $n$ copies of $p_{XYZ}$, Alice and Bob extract their common information from each copy of $p_{XYZ}$.  This will generate a sequence of $J_{XY}^n$, with Alice and Bob having identical copies of this sequence.  It is now a matter of performing privacy amplification on this sequence to remove Eve's information \cite{Bennett-1995a}.  The main construction is guaranteed to exist by the following lemma.
\begin{lemma}[See Corollary~17.5 in \cite{Csiszar-2011a}]
\label{Lem:PrivacyAmp}
For an i.i.d. source of two random variables $J_{XY}$ and $Z$ with $J_{XY}$ ranging over set $\mc{J}$, for any $\delta>0$ and $k<2^{n[H(J_{XY}|Z)-\delta]}$, there exists an $\epsilon>0$ and a mapping $\kappa:\mc{J}^n\to\mc{K}=\{1,2,\cdots ,k\}$ such that
\[\log|\mc{K}|-H(\kappa(J_{XY}^n)|Z^n)<2^{-n\epsilon}.\]
\end{lemma}
\noindent From this lemma, it follows that $H(J_{XY}|Z)$ is an achievable key rate.

\medskip
\noindent\textbf{Converse:}  The converse proof follows analogously to the converse proof of Theorem 2.6 in Ref. \cite{Csiszar-2000a} (see also \cite{Csiszar-2011a}).  We will first prove the converse under the assumption of no local randomness (i.e. $Q_A$ and $Q_B$ are constant).  We will then show that adding local randomness does not change the result.  Suppose that $\KCR=R$.  We consider a slightly weaker security condition than the one presented in Sect. \ref{Sect:Defns}.  This is done by replacing (ii) with (ii'): $\frac{1}{n}(\log|\mc{K}|-H(K|Z^nW))<\epsilon$.  Under the weaker condition, (i) implies that  
\begin{align}
\frac{1}{n}|H(f|Z^nW)-H(K|Z^nW)|&\leq \frac{1}{n}\max\{H(f|KZ^nW),H(K|fZ^nW)\}\notag\\
&\leq \frac{1}{n}\max\{H(f|K),H(K|f)\}\notag\\
&\leq\frac{1}{n} \left(h(\epsilon)+\epsilon(\log|\mc{K}|-1)\right),
\end{align}
where the last line follows from Fano's Inequality.  Hence, under the assumption of the original security condition, $\frac{1}{n}(\log|\mc{K}|-H(f|Z^nW))<\epsilon+O(\frac{\epsilon}{n})$.  This means that, without loss of generality, $K$ can be assumed to be a function of $(X^n,Q_A,W)$; i.e. $K=f(X^n,Q_A,W)$.  Then, for every $\delta,\epsilon>0$ and $n$ sufficiently large, there exists a random variable $W$ independent of $X^nY^nZ^n$ along with functions $f(X^n,W)$ and $g(Y^n,W)$ satisfying (i) $Pr[f=g=K]>1-\epsilon$, (ii') $\tfrac{1}{n}(\log|\mc{K}|-H(K|Z^nW))<\epsilon$ and (iii) $\frac{1}{n}\log|\mc{K}|\geq R$.

Note that from (i) in the security condition, Fano's Inequality together with data processing gives 
\begin{align}\label{eq_fano_00}
H(K|Y^nW)<h(\epsilon)+ \epsilon(\log|\mc{K}|-1).
\end{align}
Combining this with (ii') gives
\[\frac{1}{n}(1-\epsilon)\log|\mc{K}|<\frac{1}{n}[H(K|Z^nW)-H(K|Y^nW)+h(\epsilon)-\epsilon],\]
and so
\begin{align}
R\leq \frac{1}{n}\log|\mc{K}|+\delta<\frac{1}{1-\epsilon}\cdot\frac{1}{n}[H(K|Z^nW)-H(K|Y^nW)]+\frac{h(\epsilon)-\epsilon}{1-\epsilon}\cdot\frac{1}{n}+\delta.
\end{align}
To analyze the quantity $H(K|Z^nW)-H(K|Y^nW)$, we will use a standard trick. 

\begin{lemma}\label{lemma_key_equality}
Let $J$ be uniformly distributed over the set $\{1,\cdots,n\}$ and let $A^{(i)}$ denote the $i^{th}$ instance of $A$ in $A^n$.  Likewise, let $A^{(<i)}=A^{(1)}\cdots A^{(i-1)}$ and $A^{(>i)}= A^{(i+1)} \cdots A^{^(n)}$ with $A^{(<1)}:=\emptyset$ and $A^{(n+1)}:=\emptyset$.  
Then for random variables $P$ and $Q$ and sequences of random variables $A^n, B^n$
\begin{equation}
H(P|A^nQ) - H(P|B^n Q) = n [ I(P:B^{(J)}|TQ) - I(P:A^{(J)}|TQ) ],
\end{equation}
where $T=J A^{(>J)} B^{(<J)}$
\end{lemma}
\begin{proof}
See, e.g., proof of Lemma~17.12 in \cite{Csiszar-2011a}.
\end{proof}
\noindent Then we can use Lemma~\ref{lemma_key_equality} to obtain
\begin{align}
H(K|Z^nW)-H(K|Y^nW)&=n[I(K:Y^{(J)}|UW)-I(K:Z^{(J)}|UW)],
\end{align}
where $U:=JY^{(<J)}Z^{(>J)}$.  Notice that for any $i\in\{1,\cdots,n\}$ we have
\begin{align}
&X^{(<i)}X^{(>i)} Y^{(<i)}Z^{(>i)} -X^{(i)}-Y^{(i)}Z^{(i)},
\end{align}
since the sampling is i.i.d..  Therefore, because $K$ is a function of $(X^n,W)$, we have
\begin{align}
KU-X^{(J)}W-&Y^{(J)}Z^{(J)}.
\end{align}
Removing the superscript ``$J$'' and taking $\epsilon,\delta\to 0$, we have the bound 
\begin{align}
\label{Eq:ConverseRate}
R\leq I(K:Y|UW)-I(K:Z|UW)
\end{align}
such that $KU-XW-YZ$.  

Next, Eq.~(\ref{eq_fano_00}) gives
\begin{align}
h(\epsilon)+\epsilon(\log|\mc{K}|-1)&>H(K|Y^nW)-H(K|X^nW)\notag\\
&=n[I(K:X^{(J)}|JY^{(<J)}X^{(>J)}W)-I(K:Y^{(J)}|JY^{(<J)}X^{(>J)}W)],
\end{align}
where the first inequality follows because $H(K|X^nW)$ is nonnegative and the quality follows from Lemma~\ref{lemma_key_equality}.
We want to put this in terms of $U$.  To do this, note that
\begin{align}
I(K:X^{(J)}|JY^{(<J)}X^{(>J)}W)&=I(KY^{(<J)}X^{(>J)}:X^{(J)}|JW)\notag\\
&=I(KY^{(<J)}X^{(>J)}Z^{(>J)}:X^{(J)}|JW)-I(Z^{(>J)}:X^{(J)}|JKY^{(<J)}X^{(>J)}W) \notag\\
&=I(KUX^{(>J)}:X^{(J)}|JW)\notag\\
&=I(KU:X^{(J)}|JW)+I(X^{(>J)}:X^{(J)}|KUW),
\end{align}
where the first equality follows from the chain rule and $I(Y^{(<J)}X^{(>J)}:X^{(J)}|JW)=0$, and in the second equality 
\begin{align}
I(Z^{(>J)}:X^{(J)}|JKY^{(<J)}X^{(>J)}W)&\leq I(Z^{(>J)}:KX^{(J)}|JY^{(<J)}X^{(>J)}W)\notag\\
&=I(Z^{(>J)}:X^{(J)}|JY^{(<J)}X^{(>J)}W)\label{eq_zx}\\
&=0. \notag
\end{align}
The first equality (\ref{eq_zx}) uses $I(Z^{(>J)}:K|JY^{(<J)}X^{(\geq J)}W)=0$ since $K - JY^{(<J)}X^{(\geq J)}W - Z^{(>J)}$ is a Markov chain. Again this follows from the basic Markov condition $K-WX^n-Y^nZ^n$ and the sampling is i.i.d.. The second equality follows from i.i.d. sampling and $W$ independence of $X^n,Y^n,Z^n$.

A similar analysis likewise gives
\begin{align}
I(K:Y^{(J)}|JY^{(<J)}X^{(>J)}W)&=I(KU:Y^{(J)}|JW)+I(X^{(>J)}:Y^{(J)}|KUW)\notag\\
&\leq I(KU:Y^{(J)}|JW)+I(X^{(>J)}:X^{(J)}|KUW),
\end{align}
where the inequality follows from the Markov condition 
\[X^{(>J)}-KUX^{(J)}W-Y^{(J)},\]
which can be derived from the more obvious Markov condition
\[KUX^n-JX^{(J)}W-Y^{(J)}.\]
Putting everything together yields
\begin{align}
h(\epsilon)+\epsilon(\log|\mc{K}|-1)&>H(K|Y^nW)-H(K|X^nW)\notag\\
&>I(KU:X^{(J)}|JW)-I(KU:Y^{(J)}|JW)\notag\\
&=I(KU:X^{(J)}Y^{(J)}|JW)-I(KU:Y^{(J)}|JX^{(J)}W)-I(KU:Y^{(J)}|JW)\label{eq_f03}\\
&=I(KU:X^{(J)}|JY^{(J)}W)+I(KU:Z^{(J)}|JY^{(J)}X^{(J)}W)\label{eq_f04}\\
&=I(KU:X^{(J)}Z^{(J)}|JY^{(J)}W), \notag
\end{align}
where the second term in (\ref{eq_f03}) is zero from the already proven Markov chain $KU-XW-YZ$, and  in (\ref{eq_f04}) {we use the fact that $I(KU:Z^{(J)}|JY^{(J)}X^{(J)}W)=0$}.  Removing the superscript ``$J$'' and taking $\epsilon\to 0$ necessitates the Markov chain $KU-YW-XZ$.

The double Markov chain $K-XW-Y$ and $K-YW-X$ implies that $I(K:XY|J_{XY}W)=0$ (see Proposition \ref{Prop:Double-Markov} below).  Since $K$ is a function of $(X,W)$, we have that $H(K|J_{XY}W)=0$.  Thus, $K$ must also be a function of $(Y,W)$.  Continuing Eq. \eqref{Eq:ConverseRate} gives the bound
\begin{align}
\label{Eq:Double-Markov-Use}
R&\leq I(K:Y|UW)-I(K:Z|UW)\notag\\
&=H(K|UW)-I(K:Z|UW)\notag\\
&=H(K|ZUW)\leq H(K|ZW).
\end{align}
We have therefore obtained the following:
\begin{equation}
\label{Eq:Lem-CR-Converse}
R\leq\max H(K|ZW),
\end{equation}
where the maximization is taken over all variables $K$ such that $H(K|XW)=H(K|YW)=0$.

This can be further bounded by using the following proposition.
\begin{proposition}
\label{Prop:Common-function-independent}
If $W$ is independent of $XY$ and $H(K|XW)=H(K|YW)=0$, then $K$ is a function of $(J_{XY},W)$.
\end{proposition}
\begin{proof}
The fact that $H(K|XW)=H(K|YW)=0$ implies the existence of two functions $f(X,W)$ and $g(Y,W)$ such that $Pr[f(X,W)=g(Y,W)]=1$.  Consequently, if $p(x_1,y_1)p(x_1,y_2)>0$, then $f(x_1,w)=g(y_1,w)=g(y_2,w)$ for all $w\in\mc{W}$ with $p(w)>0$.  Indeed, if, say, $f(x_1,w)\not=g(y_1,w)$, then $Pr[f(X,W)\not=g(Y,W)]\geq p(x_1,y_1,w)=p(x_1,y_2)p(w)>0$, where we have used the independence between $XY$ and $W$.  By the same reasoning, $p(x_1,y_1)p(y_1,x_2)>0$ implies that $f(x_1,w)=f(x_2,w)=g(y_1,w)$ for all $w\in\mc{W}$.  Turning to Proposition \ref{Prop:Ergodic}, if $J_{XY}(x)=J_{XY}(x')$, then there exists a sequence $xy_1x_1y_2x_2\cdots y_n x'$ such that $p(xy_1)p(y_1x_1)p(x_1y_2)\cdots p(y_nx')>0$.  Therefore, as just argued, we must have that $f(x,w)=f(x',w)$ for all $w\in \mc{W}$.  Hence $K$ must be a function of $(J_{XY},W)$.
\end{proof}
We now apply Proposition \ref{Prop:Common-function-independent} to Eq. \eqref{Eq:Lem-CR-Converse}.  Suppose that $K$ obtains the maximization in Eq. \eqref{Eq:Lem-CR-Converse}.  Then, since $K$ is a function of $(J_{XY},W)$, we have that 
\begin{equation}
H(K|ZW)\leq H(J_{XY}W|ZW)=H(J_{XY}|ZW)\leq H(J_{XY}|Z).
\end{equation}
This proves the desired upper bound under no local randomness.

To consider the case when Alice and Bob have local randomness $Q_A$ and $Q_B$, respectively, define $\hat{X}:=(X,Q_A)$ and $\hat{Y}:=(Y,Q_B)$.  Then repeating the above argument shows that $R\leq H(J_{\hat{X}\hat{Y}}|Z)$.  It is straightforward to show that with $Q_A$ and $Q_B$ pairwise independent and independent of $XY$, we have $J_{\overline{X},\overline{Y}}=J_{XY}$.

We complete the proof by giving the Double Markov Chain Proposition used to obtain equation \eqref{Eq:Double-Markov-Use} above.
\begin{proposition}[Conditional Double Markov Chains (also Exercise 16.25 in \cite{Csiszar-2011a})]
\label{Prop:Double-Markov}
Random variables $WXYZ$ satisfy the two Markov chains $X-YZ-W$ and $Y-XZ-W$ iff $I(XY:W|J_{XY|Z}Z)=0$.
\end{proposition}

\begin{proof}
If $I(XY:W|J_{XY|Z}Z)=0$ then $I(Y:W|J_{XY|Z}Z)=0$.  The Markov chain $X-YZ-W$ follows since 
\begin{align*}
I(XY:W|J_{XY|Z}Z)&=I(X:W|YJ_{XY|Z}Z)+I(Y:W|J_{XY|Z}Z)\notag\\
&=I(X:W|YZ)+I(Y:W|J_{XY|Z}Z),
\end{align*}
where we have use the fact that $J_{XY|Z}$ is a function $X$ and $Y$ when given $Z$.  A similar argument shows that $Y-XZ-W$.

On the other hand, if the two Markov chains hold, then whenever $p_{XYZ}{x,y,z}>0$, we have
\begin{equation}
p(W=w|x,y,z)=p(w|x,z)=p(w|y,z).
\end{equation}
Hence, the conditional distribution $p(w|x,y,z)$ is constant across each block $\mc{X}_i\times\mc{Y}_i$ in the maximal common partitioning of $P_{XY|Z=z}$.  Consequently,
\[p_{W|XYZ}=p_{W|J_{XY|Z}Z},\]
and so for any $J_{XY|Z}=j$ and $Z=z$ for which $p(j,z)>0$, we have
\begin{align}
p(x,y,w|j,z)&=p(w|x,y,j,z)p(x,y|j,z)\notag\\
&=p(w|x,y,z)p(x,y|j,z)=p(w|j,z)p(x,y|j,z).
\end{align}
Thus, $I(XY:W|J_{XY|Z}Z)=0$.
\end{proof}
\end{proof}

In Ref. \cite{Chitambar-2014d} we have studied a related quantity known as the \textit{maximal conditional common function} $J_{XY|Z}$, which is the collection of variables $\{J_{XY|Z=z}:z\in\mc{Z}\}$ with $J_{XY|Z=z}$ being a maximal common function of the conditional distribution $p_{XY|Z=z}$.  The variable $J_{XY|Z}$ is again unique for every distribution $p_{XYZ}$ up to relabeling.  Since $J_{XY|Z=z}$ is computed from both $X$ and $Y$ with the additional information that $Z=z$, maximality of $J_{XY|Z=z}$ ensures that $J_{XY}$ is a function of $J_{XY|Z=z}$ for each $z\in\mc{Z}$.  In other words, a labeling of $J_{XY}$ and $J_{XY|Z}$ can be chosen so that $J_{XY}$ is a coarse-graining of $J_{XY|Z}$.  Therefore, $H(J_{XY}|Z)\leq H(J_{XY|Z}|Z)$ with equality iff $H(J_{XY|Z}|Z J_{XY})=0$.  When the equality condition holds, it means that for each $z\in\mc{Z}$, the value of $J_{XY|Z=z}$ can be determined from $J_{XY}$ alone.  Hence, the variables $J_{XY}$ and $J_{XY|Z}$ must be equivalent up to relabeling.  From this it follows that a distribution satisfies $H(J_{XY|Z}|Z J_{XY})=0$ iff it admits a decomposition of
\begin{align}
\label{Eq:UB-structure}
p(x,y,z)=\sum_{J_{XY}=j}p(x,y|z,j)p(j|z)p(z),
\end{align}
where for any $x,x'\in\mc{X}$, $y,y'\in\mc{Y}$ and $z,z'\in\mc{Z}$ the conditional distributions satisfy 
\begin{align*}
p(x,y|z,j)p(x,y'|z',j')&=0,& p(x,y|j)p(x',y|z',j')&=0\quad\text{if}\quad j\not=j'.
\end{align*}
The class of distributions of this form we shall call \textit{uniform block} (UB) (see Fig. \ref{Fig:Uniform_Block}).

\begin{figure}[t]
   \centering
     \includegraphics[width=.7\textwidth]{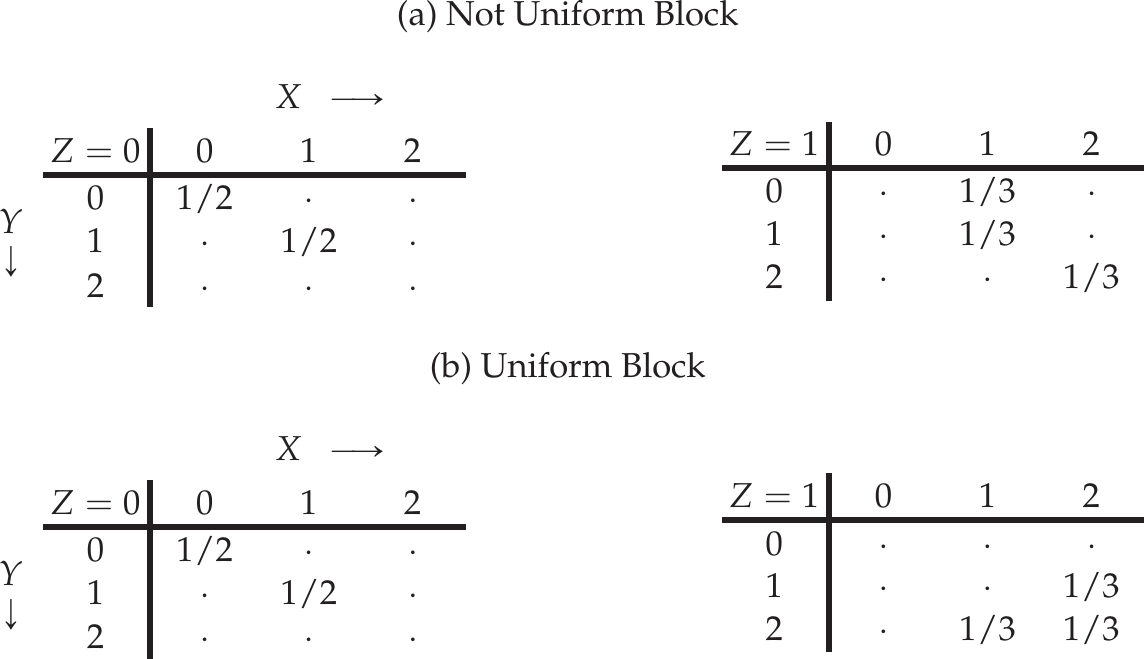}
    \captionsetup{font={small}}
    \caption[.]{Examples of a distribution that is not uniform block (a) and one that is (b).  Each entry corresponds to a conditional probability value $p(x,y|z)$.  UB distribution (b) is not uniform block independent (UBI) since the block in the $Z=1$ plane contains correlations between Alice and Bob.}
       \label{Fig:Uniform_Block}
\end{figure}

The quantity $H(J_{XY|Z}|Z)$ is the private key rate when Eve is helping by announcing her variable, yet Alice and Bob are still prohibited from communicating with one another.  Thus, the difference $H(J_{XY|Z}|Z)-H(J_{XY}|Z)$ quantifies how much Eve can assist Alice and Bob in distilling key when no communication is exchanged between the two.  From the previous paragraph, it follows that Eve offers no assistance (i.e. the private key rate equals the secret key rate) in the no-communication scenario iff the distribution is UB.

Returning to Theorem \ref{Thm:CR-Theorem}, we can now answer the underlying question of this paper for no-communication distillation.  By using the chain rule of conditional mutual information and the fact that $J_{XY}$ is both a function of $X$ and $Y$, we readily compute
\begin{align}
I(X:Y|Z)=I(J_{XY}X:Y|Z)&=I(J_{XY}:Y|Z)-I(X:Y|ZJ_{XY})&\notag\\
&=H(J_{XY}|Z)-I(X:Y|ZJ_{XY}).
\end{align}
The conditional mutual information is thus an achievable rate whenever $I(X:Y|ZJ_{XY})=0$.  Distributions satisfying this equality are uniform block with the extra condition that $p(x,y|z,j)=p(x|z,j)p(y|z,j)$ in Eq. \eqref{Eq:UB-structure}.  We shall call distributions having this form \textit{uniform block independent} (UBI).  Putting everything together, we find that 
\begin{corollary}
\label{Cor:CR-Corollary}
A distribution $p_{XYZ}$ satisfies $K^{c.r.}(X:Y||Z)=I(X:Y|Z)$ if and only if it is uniformly block independent.  
\end{corollary}

\begin{remark}
The no-communication results discussed above and proven in the appendix are already implicit in the work of Csisz\'{a}r and Narayan.  In Ref. \cite{Csiszar-2000a}, they study various key distillation scenarios with Eve functioning as a helper and limited communication between Alice and Bob.  Included in this is the no-communication scenario with and without helper.  However, being very general in nature, Csisz\'{a}r and Narayan's results involve optimizations over auxiliary random variables, and it is therefore still a non-trivial matter to discern Theorem \ref{Thm:CR-Theorem} and Corollary \ref{Cor:CR-Corollary} directly from their work.  Additionally, they do not consider the scenario of just shared public randomness.
\end{remark}

\subsection{Obtaining $I(X:Y|Z)$ with One-Way Communication}

\label{Sect:One-Way}

In this section we want to identify the type of tripartite distributions from which secret key can be distilled at the rate $I(X:Y|Z)$ using one-way communication.  Since $K(X:Y|Z)\leq I(X:Y|Z)$, our analysis deals with distributions for which one-way communication suffices to optimally distill secret key.  Manipulating Eq. \eqref{Eq:One-way Rate} of Lemma \ref{Lem:AC-Lemma} allows us to determine when $\KAB=I(X:Y|Z)$.  We have that
\begin{align}
I(K:Y|U)-I(K:Z|U)&=I(K:Y|ZU)-I(K:Z|YU)\notag\\
&=I(KU:Y|Z)-I(U:Y|Z)-I(K:Z|YU)\notag\\
&=I(X:Y|Z)-I(X:Y|KUZ)-I(U:Y|Z)-I(K:Z|YU).
\end{align}
From this and Lemma \ref{Lem:AC-Lemma}, we conclude the following.
\begin{lemma}
\label{Lem:Markov-conds}
Distribution $p_{XYZ}$ has $\overrightarrow{K}(X:Y||Z)=I(X:Y|Z)$ iff there exists variables $KUXYZ$ with $K$ and $U$ ranging over sets of size no greater than $|\mc{X}|+1$ such that
\begin{align}
\label{Eq:EqualityConds}
&(1)\quad KU-X-YZ,&&(2)\quad X-KUZ-Y,\notag\\
&(3)\quad U-Z-Y,&&(4)\quad K-YU-Z.
\end{align}
\end{lemma} 

The conditions of Lemma \ref{Lem:Markov-conds} allow for the follow rough interpretation.  (1) says that Alice is able to generate variables $K$ and $U$ from knowledge of her variable $X$.  We think of $K$ as containing the key that Alice and Bob will share and $U$ as the public message sent from Alice to Bob.  (2) says that from Eve's perspective, Alice and Bob share no more correlations given $U$ and $K$.  Likewise, (3) says that from Eve's perspective, the public message is uncorrelated with Bob.  Finally, (4) says that after learning $U$, Bob can generate the key $K$ that is independent from Eve.

Unfortunately, Lemma \ref{Lem:Markov-conds} does not provide a transparent characterization of the distributions for which $\KAB=I(X:Y|Z)$.  We next proceed to obtain a better picture of these distributions by exploring additional consequences of the Markov chains in Eq. \eqref{Eq:EqualityConds}.  The following places a necessary condition on the distributions.  We will see in Section \ref{Sect:CC-complexity}, however, that it fails to be sufficient.

\begin{theorem}
\label{Thm:One-way-conditions1}
If distribution $p_{XYZ}$ has either $\KAB=I(X:Y|Z)$ or $\KBA=I(X:Y|Z)$, then $p_{XYZ}$ must have the following property: For any $z\in\mc{Z}$, if $\mc{X}_i\times\mc{Y}_i$ and $\mc{X}_j\times\mc{Y}_j$ are two distinct blocks in the maximal common partitioning of $p_{XY|Z=z}$, then 
\[p_{XY}(\mc{X}_i,\mc{Y}_j)=0.\]
\end{theorem}
\begin{proof}
Without loss of generality, assume that $\KAB=I(X:Y|Z)$.  For distribution $p_{XY|Z=z}$ with maximal common partition $(\mc{X}_\lambda,\mc{Y}_\lambda)_{\lambda=1}^t$, consider arbitrary $(x_i,y_i)\in\mc{X}_i\times\mc{Y}_i$ and $(x_j,y_j)\in\mc{X}_j\times\mc{Y}_j$.  Note that from the definition of a maximal common partitioning, we have that $p(x_i,z)p(y_i,z)>0$, but we need not have that $p(x_i,y_i,z)>0$. 

We will prove that $p(x_i,y_j,z')=0$ for all $z'\in\mc{Z}$ (clearly this already holds when $z'=z$).  Suppose on the contrary that $p(x_i,y_j,z')>0$.  Since $p(x_i,z)>0$, there will exist some $y_i'\in\mc{Y}_i$ such that $p(x_i,y_i',z)>0$.  Then the Markov chain condition $KU-X-YZ$ implies that for some $(k,u)\in\mc{K}\times\mc{U}$ such that $p(k,u|x_i)>0$, we have
\begin{equation}
\label{Eq:Markov1}
p(k,u|x_i)=p(k,u|x_i,y_i',z)=p(k,u|x_i,y_j,z')>0.
\end{equation}
Eq. \eqref{Eq:Markov1} implies that both $p(k,u|y_i',z)>0$ and $p(k,u|y_j,z')>0$.  From $p(u|y_i',z)>0$ and the Markov chain $U-Z-Y$, we have that $p(u|y_j,z)>0$.  Then we can further derive
\begin{align}
0<p(k,u|y_j,z')&=p(u|y_j,z')p(k|u,y_j,z')\notag\\
&=p(u|y_j,z')p(k|u,y_j,z)\notag\\
&\quad\Rightarrow\quad p(k|u,y_j,z)>0,\notag\\
&\quad\Rightarrow\quad p(k,u|y_j,z)=p(k|u,y_j,z)p(u|y_j,z)>0,
\end{align}
where we have used the Markov chain $K-YU-Z$.  From the last line, we must be able to find some $x_j'\in\mc{X}_j$ such that $p(x_j',y_j,z)>0$ and $p(k,u|x_j',y_j,z)>0$.  Inverting probabilities gives that both $p(x_j',y_j|k,u,z)>0$ and $p(x_i,y_i'|k,u,z)>0$.  Hence,
\begin{align}
I(X:Y|KUZ)&= I(J_{XY|Z}X:Y|KUZ)\notag\\
&=I(X:Y|J_{XY|Z}KUZ)+ \sum_{k,u,z}H(J_{XY|Z=z}|k,u,z)p(k,u,z)>0,
\end{align}
since $H(J_{XY|Z=z}|k,u,z)>0$ because $(x_i,y_i')\in\mc{X}_i\times\mc{Y}_i$ and $(x_j',y_j)\in\mc{X}_j\times\mc{Y}_j$.  However, this strict inequality contradicts the Markov chain condition $X-KUZ-Y$.
\end{proof}

Figure \ref{Fig:One-way_Examples} (a) provides an example distribution which does not satisfy the necessary conditions of Theorem \ref{Thm:One-way-conditions1} for $I(X:Y|Z)$ to be an achievable one-way key rate.  On the other hand, Figure \ref{Fig:One-way_Examples} (b) depicts an distribution for which the conditions of the theorem are met.  However, Theorem \ref{Thm:DistOpt1} in the next section will show that both distributions (a) and (b) have $K(X:Y||Z)<I(X:Y|Z)$.

\begin{figure}[t]
   \centering
     \includegraphics[width=.75\textwidth]{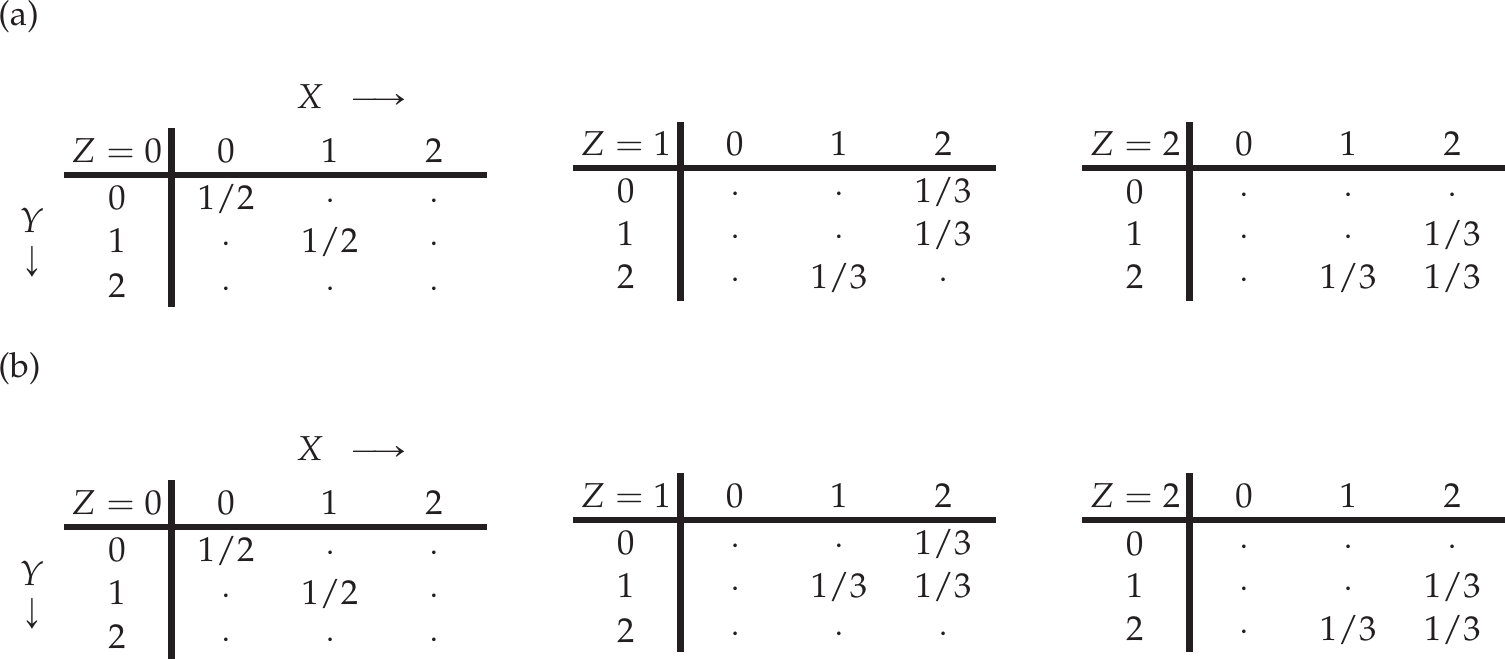}
 		\captionsetup{font={small}}
    \caption[.]{(a) The conditions for a one-way key rate of $I(X:Y|Z)$ given by Theorem \ref{Thm:One-way-conditions1} are violated for this distribution.  To see this, note that the events $(X=1,Y=2)$ and $(X=2,Y=1)$ are both possible when $Z=1$.  Hence, Theorem \ref{Thm:One-way-conditions1} necessitates $p(1,1)=0$, which is not the case because of the plane $Z=0$.  Distribution (b) lacks this characteristic and therefore it satisfies the conditions of Theorem \ref{Thm:One-way-conditions1}.}
       \label{Fig:One-way_Examples}
\end{figure}

\subsection{Obtaining $I(X:Y|Z)$ with Two-Way Communication}
\label{Sect:Two-Way}

We now turn to the general scenario of interactive two-way communication.  Our main result is the necessary structural condition of Theorem \ref{Thm:DistOpt1}.  Its statement requires some new terminology.

For two distributions $p_{XY}$ and $q_{XY}$ over $\mc{X}\times \mc{Y}$, we say that $q_{XY}\blacktriangleleft p_{XY}$ if, up to a permutation between $X$ and $Y$, the distributions satisfy $supp[q_X]\subset supp[p_X]$ and one of the three additional conditions: (i) $q_{XY}$ is uncorrelated, (ii) $supp[q_Y]\subset supp[p_Y]$, or (iii) $y\in supp[q_Y]\setminus supp[p_Y]$ implies that $H(X|Y=y)=0$. 

\begin{theorem}
\label{Thm:DistOpt1}
Let $p_{XYZ}$ be a distribution over $\mc{X}\times\mc{Y}\times\mc{Z}$ such that $p_{XY|Z=z_1}\blacktriangleleft p_{XY|Z=z_0}$ for some $z_0,z_1\in\mc{Z}$.  If there exists some pair $(x,y)\in supp[p_{X|Z=0}]\times supp[p_{Y|Z=0}]$ for which $p(x,y|z_1)>0$ but $p(x,y|z_0)=0$, then $K(X:Y||Z)<I(X:Y|Z)$.  
\end{theorem}
\begin{proof}
The proof will involve showing that there exists a channel $\overline{Z}|Z$ such that $I(X:Y|\overline{Z})<I(X:Y|Z)$.  The channel will involve mixing $z_0$ and $z_1$ but leaving all other elements unchanged.  Define the function
\begin{equation}
\label{Eq:InfoMix}
f(t)=I(X:Y)_{(1-t)p_{XY|Z=z_0}+t p_{XY|Z=z_1}}\qquad t\in[0,1],
\end{equation}
which gives the mutual information of the mixed distribution $(1-t)p_{XY|Z=z_0}+tp_{XY|Z=z_1}$.  The function $f$ is continuous and twice differentiable in the open interval $(0,1)$.  To prove the theorem, we will need a simple general fact about functions of this sort.
\begin{proposition}
\label{Prop:calc}
Suppose that $f$ is a continuous function on the closed interval $[0,1]$ and twice differentiable in the open interval $(0,1)$.  Suppose there exists some $0<\delta<1$ such that $f$ is strictly convex in the interval $\mc{I}=(0,\delta]$ and $f(1)-f(0)>f'(t)$ for all $t\in\mc{I}$.  Then $f(t)<(1-t)f(0)+tf(1)$ for all $t\in\mc{I}$.
\end{proposition}
\begin{proof}
Introduce the linear function $g(t)=(1-t)f(0)+tf(1)$.  Note that by assumption we have $g'(t)>f'(t)$ for $t\in\mc{X}$.  We want to show that $f(t)<g(t)$ for $t\in \mc{I}$.  We have
\begin{align}
g(t)&=(1-\tfrac{t}{\delta})g(0)+\tfrac{t}{\delta}g(\delta)>(1-\tfrac{t}{\delta})f(0)+\tfrac{t}{\delta}f(\delta)>f(t).
\end{align}
Here, the first inequality follows from the facts that $f(0)=g(0)$ and $0>g'(t)>f'(t)$ for $t\in\mc{I}$ (so $g(\delta)>f(\delta)$); and the second inequality uses the strict convexity of $f$ in $\mc{I}$.
\end{proof}

Continuing with the proof of Theorem \ref{Thm:DistOpt1}, it will suffice to show that the function given by Eq. \eqref{Eq:InfoMix} satisfies the conditions of Proposition \ref{Prop:calc}.  For if this is true, then we can argue as follows.  Choose $\epsilon$ sufficiently small so that $\tfrac{\epsilon p(z_1)}{p(z_0)+\epsilon p(z_1)}\in(0,\delta]$, where $\delta$ is described by the proposition.  Define the channel $\overline{Z}|Z$ by  $p(\overline{z}_0|z_1)=\epsilon$, $p(\overline{z}_1|z_1)=1-\epsilon$, and $p(\overline{z}|z)=1$ for all $z\not=z_1\in\mc{Z}$.  This means that $p(\overline{z}_0)=p(z_0)+\epsilon p(z_1)$ and $p(\overline{z}_1)=(1-\epsilon)p(z_1)$, and inverting the probabilities gives $p(z_1|\overline{z}_1)=1$, $p(z_1|\overline{z}_0)=\tfrac{\epsilon p(z_1)}{p(z_0)+\epsilon p(z_1)}$, and $p(z_0|\overline{z}_0)=\tfrac{p(z_0)}{p(z_0)+\epsilon p(z_1)}$.  Since $p(x,y|\overline{Z}=\overline{z})=\sum_zp(x,y|Z=z)p(Z=z|\overline{Z}=\overline{z})$, the average conditional mutual information is 
\begin{align}
&\sum_{z\not=z_0,z_1\in\mc{Z}}I(X:Y|\overline{Z}=\overline{z})p(\overline{z})+f(\tfrac{\epsilon p(z_1)}{p(z_0)+\epsilon p(z_1)})p(\overline{z}_0)+f(1)p(\overline{z}_1)\notag\\
&< \sum_{z\not=z_0,z_1\in\mc{Z}}I(X:Y|Z=z)p(z)+\left(\tfrac{p(z_0)}{p(z_0)+\epsilon p(z_1)}f(0)+\tfrac{\epsilon p(z_1)}{p(z_0)+\epsilon p(z_1)} f(1)\right)p(\overline{z}_0)+f(1)(1-\epsilon)p(z_1)\notag\\
&=I(X:Y|Z),
\end{align}
where Proposition \ref{Prop:calc} at $x=\tfrac{\epsilon p(z_1)}{p(z_0)+\epsilon p(z_1)}$ has been invoked.

Let us then show that the conditions of Proposition \ref{Prop:calc} hold true for the function given by Eq. \eqref{Eq:InfoMix} whenever $p_{XY|Z=z_1}\blacktriangleleft p_{XY|Z=z_0}$; i.e. that there exists some interval $(0,\delta]$ for which $f$ is strictly convex and $f(1)-f(0)>f'(t)$.  We have
\begin{align}
\label{Eq:MI-combo1}
f(t)=&-\sum_{x\in\mc{X}}[(1-t)p(x|z_0)+tp(x|z_1)]\log [(1-t)p(x|z_0)+tp(x|z_1)]\notag\\
&-\sum_{y\in\mc{Y}}[(1-t)p(y|z_0)+tp(y|z_1)]\log [(1-t)p(y|z_0)+tp(y|z_1)]\notag\\
&+\sum_{x\in\mc{X}}\sum_{y\in\mc{Y}} [(1-t)p(x,y|z_0)+tp(x,y|z_1)]\log [ (1-t)p(x,y|z_0)+tp(x,y|z_1)].
\end{align}
We are interested in $\lim_{t\to 0}f'(t)$ and $\lim_{t\to 0}f''(t)$.  To compute these, we use the fact that the function $g(t)=(r+s t)\log(r+st)$ satisfies $g'(t)=s (1 + \log(r + s t))$ and $g''(t)=\frac{s^2}{r+st}$.  We separate the analysis into three cases.  Without loss of generality, we will assume $supp[p_{X|Z=z_1}]\subset supp[p_{X|Z=z_0}]$.

\medskip

\noindent\textbf{Case (i): $\boldsymbol{p_{XY|Z=z_1}}$ is uncorrelated.}

Since $supp[p_{X|Z=z_1}]\subset supp[p_{X|Z=z_0}]$, we can assume that $p(x|z_0)\not=0$ for all $x$; otherwise there is no term involving $x$ in Eq. \eqref{Eq:MI-combo1}.  Now suppose that $p(y|z_0)=0$.  Then for this fixed $y$, the summation over $x$ in the third term of Eq. \eqref{Eq:MI-combo1} becomes
\begin{align}
&\sum_{x\in\mc{X}} [(1-t)p(x,y|z_0)+tp(x,y|z_1)]\log [ (1-t)p(x,y|z_0)+tp(x,y|z_1)]\notag\\
&=t\sum_{x\in\mc{X}} p(x|z_1)p(y|z_1)\log [ tp(x|z_1)p(y|z_1)]\notag\\
&=tp(y|z_1)\log [ tp(y|z_1)]+tp(y|z_1)\sum_{x\in\mc{X}}p(x|z_1)\log [ p(x|z_1)].
\end{align}
Hence, by letting $\mc{B}_I=\{y:p(y|z_I)>0\}$ for $I\in\{0,1\}$, we can equivalently write Eq. \eqref{Eq:MI-combo1} as
\begin{align}
\label{Eq:MI-combo2}
f(t)=&-\sum_{x\in\mc{X}}[(1-t)p(x|z_0)+tp(x|z_1)]\log [(1-t)p(x|z_0)+tp(x|z_1)]\notag\\
&-\sum_{y\in\mc{B}_0}[(1-t)p(y|z_0)+tp(y|z_1)]\log [(1-t)p(y|z_0)+tp(y|z_1)]\notag\\
&+\sum_{y\in\mc{B}_0}\sum_{x\in\mc{X}} [(1-t)p(x,y|z_0)+tp(x,y|z_1)]\log [ (1-t)p(x,y|z_0)+tp(x,y|z_1)]\notag\\
&+t\sum_{y\in\mc{B}_1\setminus\mc{B}_0}p(y|z_1)\sum_{x\in\mc{X}}p(x|z_1)\log [ p(x|z_1)].
\end{align}
If $p(x,y|z_0)=0$ for some $(x,y)\in\mc{X}\times\mc{B}_0$, then the first derivative of \eqref{Eq:MI-combo2} will diverge to $-\infty$ as $t\to 0$ while its second derivative will diverge to $+\infty$ whenever $p(x,y|z_1)>0$.  But by assumption, there is at least one pair of $(x,y)$ for which this latter case holds.  Hence, an interval $(0,\delta]$ can always be found for which Proposition \ref{Prop:calc} can be applied to $f$.

\medskip

\noindent\textbf{Case (ii): $\boldsymbol{\mc{B}_1\setminus\mc{B}_0=\emptyset}$.}  

This is covered in case (iii).

\medskip

\noindent\textbf{Case (iii): $\boldsymbol{y\in \mc{B}_1\setminus\mc{B}_0\;\Rightarrow\; p(y|z_1)=p(x_y,y|z_1)}$ for some particular $\boldsymbol{x_y\in\mc{X}}$.}

The condition $p(y|z_1)=p(x_y,y|z_1)$ implies that $p(x,y|z_1)=0$ for all $x\not=x_y$.  Then similar to the previous case, when $y\in\mc{B}_1\setminus\mc{B}_0$, the summation over $x$ in the third term of Eq. \eqref{Eq:MI-combo1} is
\begin{align}
\sum_{x\in\mc{X}} tp(x,y|z_1)\log [ tp(x,y|z_1)]&=tp(x_y,y|z_1)\log [ tp(x_y,y|z_1)]\notag\\
&=tp(y|z_1)\log [ tp(y|z_1)].
\end{align}
Hence each term with $y\in\mc{B}_1\setminus\mc{B}_0$ becomes canceled in Eq. \eqref{Eq:MI-combo1}.  Then Eq. \eqref{Eq:MI-combo1} reduces to
\begin{align}
\label{Eq:MI-combo3}
f(t)=&-\sum_{x\in\mc{X}}[(1-t)p(x|z_0)+tp(x|z_1)]\log [(1-t)p(x|z_0)+tp(x|z_1)]\notag\\
&-\sum_{y\in\mc{B}_0}[(1-t)p(y|z_0)+tp(y|z_1)]\log [(1-t)p(y|z_0)+tp(y|z_1)]\notag\\
&+\sum_{x\in\mc{X}}\sum_{y\in\mc{B}_0} [(1-t)p(x,y|z_0)+tp(x,y|z_1)]\log [ (1-t)p(x,y|z_0)+tp(x,y|z_1)].
\end{align}
As in the previous case, the first derivative of this function will diverge to $-\infty$ while its second derivative will diverge to $+\infty$ whenever $p(x,y|z_1)>0$ and $p(x,y|z_0)=0$.  By assumption, such a pair $(x,y)$ exists, and so again, an interval $(0,\delta]$ can always be found for which Proposition \ref{Prop:calc} can be applied to $f$.  Note that when $\mc{B}_1\setminus\mc{B}_0=\emptyset$, as in case (ii), Eq. \eqref{Eq:MI-combo3} is equivalent to \eqref{Eq:MI-combo1}.  The derivative argument can thus be applied directly to \eqref{Eq:MI-combo1}.
\end{proof}

Theorem \ref{Thm:DistOpt1} is quite useful in that it allows us to quickly eliminate many distributions from achieving the rate $I(X:Y|Z)$.  For example, consider when $p_{XY|Z=z}$ is uncorrelated for some $z\in\mc{Z}$, but $p_{XY|Z=z'}$ is perfectly correlated for some other $z'\in\mc{Z}$ with either $supp[p_{X|Z=z}]\subset supp[p_{X|Z=z'}]$ or $supp[p_{Y|Z=z}]\subset supp[p_{Y|Z=z'}]$.  Here, perfectly correlated means that $p(x,y|z')=p(x|z')\delta_{x,y}$ up to relabeling.  Then from Theorem \ref{Thm:DistOpt1}, it follows that $I(X:Y|Z)$ is an achievable rate only if 
\[p(x,y|z)>0\quad\Rightarrow\quad  p(x|z')p(y|z')=0.\]
In other words, it is always possible for either Alice or Bob to identify when $Z\not=z'$. 

Finally, we close this section by comparing Theorems \ref{Thm:One-way-conditions1} and \ref{Thm:DistOpt1}.  In short, neither one supersedes the other.  As noted above, distribution (b) in Fig. \ref{Fig:One-way_Examples} satisfies the necessary condition of Theorem \ref{Thm:One-way-conditions1} for $\overrightarrow{K}(X:Y||Z)=I(X:Y|Z)$.  However, Theorem \ref{Thm:DistOpt1} can be used to show that $K(X:Y||Z)<I(X:Y|Z)$.  This is because $p_{XY|Z=1}\blacktriangleleft p_{XY|Z=2}$ yet $p(1,1|2)=0$ while $p(1,1|1)=1/3$.  Therefore its key rate is strictly less than $I(X:Y|Z)$.  Figure \ref{Fig:two-way-counter} depicts a distribution for which Theorem \ref{Thm:DistOpt1} cannot be applied but Theorem \ref{Thm:One-way-conditions1} shows that $\KAB<I(X:Y|Z)$.  The two-way key rate for this distribution is still unknown.

\begin{figure}[t]
   \centering
     \includegraphics[width=.75\textwidth]{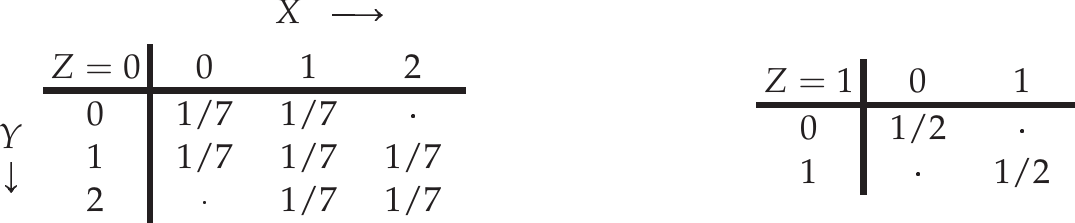}
 		\captionsetup{font={small}}
    \caption[.]{The event $(x,y)=(0,1)$ has conditional probabilities $p(0,1|Z=0)>0$ and $p(0,1|Z=1)=0$.  However, we cannot use these facts in conjunction with Theorem \ref{Thm:DistOpt1} to conclude that $K(X:Y||Z)<I(X:Y|Z)$ since the distribution does not satisfy $p_{XY|Z=0}\blacktriangleleft p_{XY|Z=1}$ (neither $supp[p_{X|Z=0}]\subset supp[p_{X|Z=1}]$ nor $supp[p_{Y|Z=0}]\subset supp[p_{Y|Z=1}]$).  On the other hand, since $p(0,1|Z=0)>0$, Theorem \ref{Thm:One-way-conditions1} can be applied to conclude that the one-way rate is less than $I(X:Y|Z)$.}
       \label{Fig:two-way-counter}
\end{figure}

\subsection{Communication Dependency in Optimal Distillation}

\label{Sect:CC-complexity}

We next consider some general features of the public communication when performing optimal key distillation.  Our main observations will be that (i) attaining a key rate of $I(X:Y|Z)$ by one-way communication may depend on the direction of the communication, and (ii) two-way communication may be necessary in order to achieve the key rate $I(X:Y|Z)$.  

\begin{figure}[b]
   \centering
     \includegraphics[width=.75\textwidth]{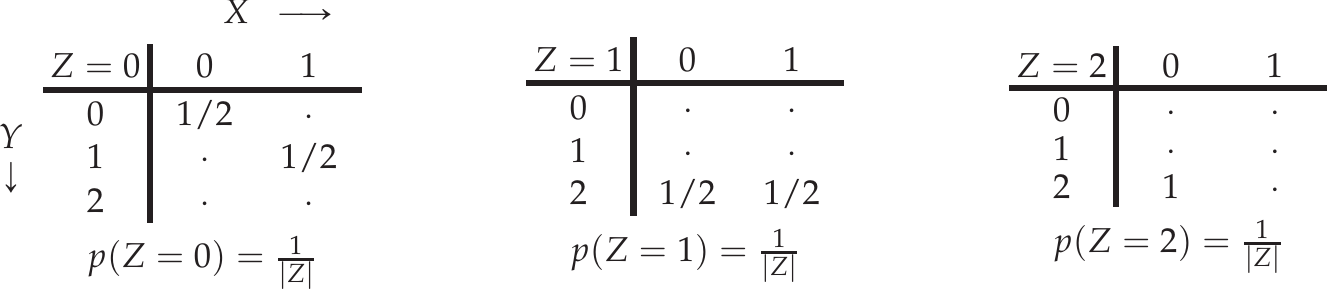}
 		\captionsetup{font={small}}
        \caption[.]{A distribution requiring communication from Bob to Alice to achieve a key rate of $I(X:Y|Z)$.}
       \label{Fig:One-way-counterexample}
\end{figure}

\begin{example}[Optimal one-way distillation depends on communication direction]
Consider the distribution depicted in Fig. \ref{Fig:One-way-counterexample} with $I(X:Y|Z)=1/3$.  When Bob is the communicating party, a protocol attaining this as a key rate is obvious: he simply announces whether or not $y\in\{0,1\}$.  If it is, they share one bit, otherwise they fail.  Hence, $I(X:Y|Z)=1/3$ is an achievable key rate.

However, the interesting question is whether or not the key rate $I(X:Y|Z)$ is achievable by one-way communication from Alice to Bob.  We will now show that this is not possible.  By Lemma \ref{Lem:Markov-conds}, in order to obtain the rate $I(X:Y|Z)$, there must exist random variables $U$ and $V$ satisfying Eq. \eqref{Eq:EqualityConds}.  Assume that such variables exist.  If $U-Z-Y$, then $p(u|X=0)p(u|X=1)>0$ for all $U=u$; otherwise, $U$ and $Y$ couldn't be independent.  But then $X-KUZ-Y$ applied to $Z=0$ means there must exist a pair $(k,u)\in\mathcal{K}\times\mathcal{U}$ such that 
\[p(k,u|X=0)=0\quad\&\quad p(k,u|X=1)>0.\]
Hence, $0=p(k|Y=2,U=u,Z=2)<p(k|Y=2,U=u,Z=1)$, which contradicts $K-YU-Z$.  Thus $\KAB<I(X:Y|Z)= \KBA$.  
\end{example}

In this example, notice that if we restricted Eve's distribution to $\mc{Z}=\{0,1\}$ (i.e $p(Z=2)=0$), then the rate $I(X:Y|Z)$ would indeed be achievable using one-way communication from Alice to Bob.  This is because without the $z=2$ outcome, the Markov Chain $X-Y-Z$ holds.  Such a result is counter-intuitive since Alice and Bob share no correlations when $z\in\{1,2\}$.  And yet the distribution becomes one-way reversible from Alice to Bob when $p(Z=2)=0$, but otherwise it is not.

\begin{figure}[t]
   \centering
     \includegraphics[width=.75\textwidth]{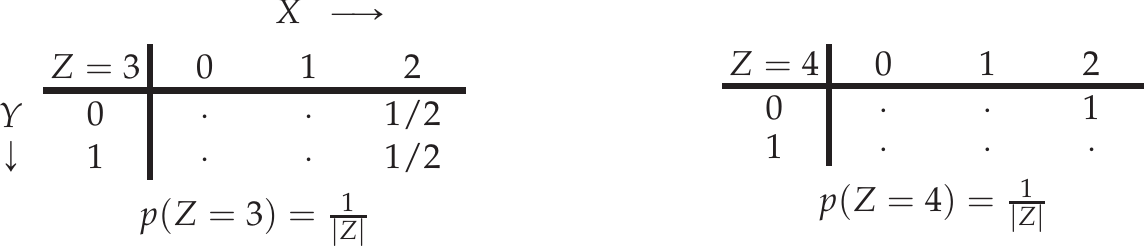}
 		\captionsetup{font={small}}
        \caption[.]{Additional outcomes augmented to the distribution of Fig. \ref{Fig:One-way-counterexample}.  The enlarged distribution can no longer attain a key rate of $I(X:Y|Z)$ unless both parties communicate.}
       \label{Fig:One-way-counterexample-2way}
\end{figure}

\begin{example}[Optimal distillation requires two-way communication]
The previous example can be generalized by adding two more outcomes for Eve so that $|Z|=5$.  The additional outcomes are shown in Fig. \ref{Fig:One-way-counterexample-2way} and this is combined with Fig. \ref{Fig:One-way-counterexample} to give the full distribution.  Notice that the distribution $p_{XY|Z=3}$ is obtained from $p_{XY|Z=1}$ simply by swapping Alice and Bob's variables, and likewise for $p_{XY|Z=4}$ and $p_{XY|Z=2}$.  Hence by the argument of the previous example, if Eve were to reveal whether or not $z\in\{0,3,4\}$, then the average Bob-to-Alice distillable key conditioned on this information would be less than $I(X:Y|Z)$.  Likewise, if Eve were to reveal whether or not $z\in\{0,1,2\}$, then the Alice-to-Bob distillable key conditioned on this information would be less than $I(X:Y|Z)$.  Thus since the average conditional key rate cannot exceed the key rate with no side information, we conclude that $I(X:Y|Z)$ is unattainable using one one-way communication in either direction.  On the other hand, the distribution is easily seen to admit a key rate of $I(X:Y|Z)$ when the parties simply announce whether or not their variable belongs to the set $\{0,1\}$.
\end{example}

\section{Conclusion}

\label{Sect:Conclusions}

In this paper, we have considered when a secret key rate of $I(X:Y|Z)$ can be attained by Alice and Bob when working with a variety of auxiliary resources.  The conditional mutual information quantifies the private key rate of $p_{XYZ}$, which is the rate of key private from Eve that is attainable when Eve helps Alice and Bob by announcing her variable.  Therefore, distributions for which $K(X:Y||Z)=I(X:Y||Z)$ are those for which no assistance is provided by Eve when she functions as a helper rather than a full adversary.  

We have found that with no additional communication, the key rate is $I(X:Y|Z)$ if and only if the distribution is uniform block independent.  Furthermore, supplying Alice and Bob with additional public randomness does not increase the distillable key rate.  While this may not be overly surprising since the considered common randomness is uncorrelated with the source, it is nevertheless a nontrivial result because in general, randomness can serve a resource in distillation tasks \cite{Ahlswede-1993a, Ozols-2014a}.

Turning to the one and two-way communication scenarios, we have presented in Theorems \ref{Thm:One-way-conditions1} and \ref{Thm:DistOpt1} necessary conditions for a distribution to attain the key rate $I(X:Y|Z)$.  The conditions we have derived are all single-letter structural characterizations, and they are thus computationally easy to apply.  We leave open the question of whether Theorem \ref{Thm:DistOpt1} is also sufficient for attaining $I(X:Y|Z)$, although we have no strong reason to believe this is true.  Further improvements to the results of this paper can possibly be obtained by studying tighter bounds on $K(X:Y||Z)$ than the intrinsic information such as those presented in Refs. \cite{Renner-2003a} and \cite{Gohari-2010a}.  Nevertheless, we hope this paper has shed new light on the problem of secret key distillation under various communication settings.  

\section{Acknowledgments}
\label{Acknowledgments}

EC was supported by the National Science Foundation (NSF) Early CAREER Award No. 1352326. MH is supported by an ARC Future Fellowship under Grant FT140100574.

\section{Appendix}

\label{Appendix}

\subsection{Proof of Propositions \ref{Prop:Partition-Unique} and \ref{Prop:Ergodic}}
  
\begin{proposition*}
\begin{itemize}
\item[{}]
\item[(a)]  Every pair of finite random variables $XY$ has a unique maximal common partitioning.
\item[(b)]  Variable $J_{XY}$ satisfies
\[H(J_{XY})=\max_K\{H(K):0=H(K|X)=H(K|Y)\}\]
iff $J_{XY}$ is a common function for the maximal common partitioning of $XY$.
\item[(c)] If $f(X)=g(Y)=C$ is any other common function of $X$ and $Y$, then $C(J_{XY})$.
\end{itemize}
\end{proposition*}
\begin{proof}
(a) Trivially $\mc{X}\times\mc{Y}$ gives a common partitioning of length one, and any common partitioning cannot have length exceeding $\min\{|\mc{X}|,|\mc{Y}|\}$; hence a maximal common partitioning exists.  To prove uniqueness, suppose that
$(\mc{X}_i,\mc{Y}_i)_{i=1}^t$ and $(\mc{X}'_i,\mc{Y}_i')_{i=1}^t$ are two maximal common partitionings.  If they are not equivalent, then there must exist some subset, say $\mc{X}_{i_0}$ such that $\mc{X}_{i_0}\subset\cup_{\lambda=1}^K\mc{X}_\lambda'$ in which $\mc{X}_{i_0}\cap\mc{X}'_{\lambda}\not=\emptyset$ for $\lambda=1,\cdots,K\geq 2$.  Choose any such $\mc{X}'_{\lambda_0}$ from this collection and define the new sets $R_{i_0}=\mc{X}_{i_0}\cap\mc{X}'_{\lambda_0}$ and $\tilde{R}_{i_0}=\mc{X}_{i_0}\setminus\mc{X}'_{\lambda_0}$, which are both nonempty since $k\geq 2$ and the $\mc{X}_\lambda$ are disjoint.  However, we also have the properties
\begin{align}
x\in\mc{X}_{i_0}&\Rightarrow p(\mc{Y}_{i_0}|x)=1;& x\in\mc{X}'_{\lambda_0}&\Rightarrow p(\mc{Y}'_{\lambda_0}|x)=1; \notag\\
x\not\in\mc{X}_{i_0} &\Rightarrow p(\mc{Y}_{i_0}|x)=0;&x\not\in\mc{X}'_{\lambda_0} &\Rightarrow p(\mc{Y}'_{\lambda_0}|x)=0.\notag
\end{align}
(Here we are implicitly using condition (iii) in the above definition by assuming that $p(x)>0$ thereby defining conditional distributions).  Therefore, $p(S_{i_0}|R_{i_0})=p(\tilde{S}_{i_0}|\tilde{R}_{i_0})=1$ and $p(S_{i_0}|\tilde{R}_{i_0})=p(\tilde{S}_{i_0}|R_{i_0})=0$, where $S_{i_0}=\mc{Y}_{i_0}\cap\mc{Y}'_{\lambda_0}$ and $\tilde{S}_{i_0}=\mc{Y}_{i_0}\setminus\mc{Y}'_{\lambda_0}$.  A similar argument shows that $p(R_{i_0}|S_{i_0})=p(\tilde{R}_{i_0}|\tilde{S}_{i_0})=1$ and $p(R_{i_0}|\tilde{S}_{i_0})=p(\tilde{R}_{i_0}|S_{i_0})=0$.  Hence, $(\mc{X}_i,\mc{Y}_i)_{i\not=i_0}^t\bigcup (S_{i_0},R_{i_0})\bigcup (\tilde{S}_{i_0},\tilde{R}_{i_0})$ is a common partitioning of length $t+1$.  But this is a contradiction since $(\mc{X}_i,\mc{Y}_i)_{i=1}^t$ is a maximal common decomposition.

(b)  Suppose that $K$ satisfies $0=H(K|X)=H(K|Y)$ so that $K=f(X)=g(Y)$ for some functions $f$ and $g$.  It is clear that $f$ and $g$ must be constant-valued for any pair of values taken from same block $\mc{X}_i\times\mc{Y}_i$ in the maximal common partitioning of $XY$.  Hence the maximum possible entropy of $K$ is then attained iff $f$ and $g$ take on a different value for each block in this partitioning.

(c)  Suppose that $C$ is not a function of $J_{XY}$.  Then $H(CJ_{XY})>H(J_{XY})$, which contradicts the maximality of $J_{XY}$.
\end{proof}

\begin{proposition*}
If $J_{XY}(x)=J_{XY}(x')$ for $x,x'\in J_{XY}$, then there exists a sequence of values
\[xy_1x_1y_2x_2\cdots y_n x'\]
such that $p(xy_1)p(y_1x_1)p(x_1y_2)\cdots p(y_nx')>0$.
\end{proposition*}
\begin{proof}
Define the sets
\begin{align}
S_0&=\{x\},&T_1&=\{y:p(y|S_0)>0\}\notag\\
S_1&=\{x\not\in S_0:p(x|T_1)>0\},&T_2&=\{y\not\in T_1:p(y|S_1\cup S_0)>0\}\notag\\
&\cdots,&T_n&=\{y\not\in T_{n-1}:p(y|\cup_{k=0}^{n-1} S_k)>0\},\notag\\
S_n&=\{x\not\in S_{n-1}:p(x|\cup_{k=1}^nT_k)>0\},&\cdots&.
\end{align}
Since $\mc{X}$ and $\mc{Y}$ are finite sets, there must exist some $M$ and $N$ such that $S_{M+1}=\emptyset$ and $T_{N+1}=\emptyset$.  Define $\overline{S}=\cup_{k=0}^M S_k$ and $\overline{T}=\cup_{k=1}^N T_k$.  By construction we have $p(\overline{S}|\overline{T})=p(\overline{T}|\overline{S})=1$, and since $J_{XY}(x)=J_{XY}(x')$ we must have $x,x'\in\overline{S}$.  However, again by construction, we can always find a sequence $xy_1x_1y_2x_3\cdots y_n x'$ with $x_k\in\cup_{i=0}^k S_i$ and $y_k\in\cup_{i=1}^k T_i$, and so \[p(xy_1)p(y_1x_1)p(x_1y_2)\cdots p(y_nx')>0.\]
\end{proof}

\bibliographystyle{alphaurl}
\bibliography{QuantumBib}

\end{document}